\documentclass[11pt]{article}
\usepackage{fullpage}

\usepackage{graphicx}
\usepackage{algorithmic}
\usepackage{algorithm}
\usepackage{amssymb}
\usepackage{amsmath}
\usepackage{amsthm}

\usepackage[linktocpage=true,pagebackref=true]{hyperref}

\newtheorem{theorem}{Theorem}
\newtheorem{lemma}[theorem]{Lemma}

\newtheorem{corollary}[theorem]{Corollary}

\theoremstyle{definition}

\theoremstyle{remark}

\renewcommand{\Pr}{\mathrm{\sc Pr}}
\newcommand{\E}{\mathbb{E}}

\newcommand{\opt}{\mathsf{opt}}

\newcommand{\poly}{\mathrm{poly}}
\newcommand{\calG}{\mathcal{G}}

\newcommand{\calF}{\mathcal{F}}
\newcommand{\calT}{\mathcal{T}}
\newcommand{\calH}{\mathcal{H}}
\newcommand{\calQ}{\mathcal{Q}}
\newcommand{\calR}{\mathcal{R}}
\newcommand{\calZ}{\mathcal{Z}}

\newcommand{\cost}{\mathsf{cost}}

\renewcommand{\deg}{\mathsf{deg}}
\newcommand{\indeg}{\mathsf{indeg}}
\newcommand{\outdeg}{\mathsf{outdeg}}

\renewcommand{\setminus}{-}

\newcommand{\polylog}{\mathrm{polylog}}

\begin{document}

\title{Approximating Directed Steiner Problems via Tree Embedding}

\author{Bundit Laekhanukit\thanks{
The Weizmann Institute of Science, Israel,
email:\url{bundit.laekhanukit@weizmann.ac.il}.
}.\footnote{
The work was partly done while the author was at 
McGill University, Simons Institute for the Theory of Computing
and the Swiss AI Lab IDSIA.
Partially supported by the ERC Starting Grant NEWNET 279352 and
by Swiss National Science Foundation project 200020\_144491/1.
}}











\maketitle

\begin{abstract}
Directed Steiner problems are fundamental problems in 
Combinatorial Optimization and Theoretical Computer Science.
An important problem in this genre is
the {\em $k$-edge connected directed Steiner tree} ($k$-DST) problem. 
In this problem, we are given a directed graph $G$ on $n$ vertices
with edge-costs, a root vertex $r$, a set of $h$ terminals $T$ 
and an integer $k$.
The goal is to find a min-cost subgraph $H\subseteq G$
that connects $r$ to each terminal $t\in T$ by
$k$ edge-disjoint $r,t$-paths.
This problem includes as special cases the well-known 
{\em directed Steiner tree} (DST) problem (the case $k=1$)
and
the {\em group Steiner tree} (GST) problem. 
Despite having been studied and mentioned many times in literature, 
e.g., by Feldman~et~al. [SODA'09, JCSS'12], 
by Cheriyan~et~al. [SODA'12, TALG'14]
and by Laekhanukit [SODA'14],
there was no known non-trivial approximation algorithm for $k$-DST
for $k \geq 2$ even in the special case that an input graph is 
directed acyclic and has a constant number of layers.
If an input graph is not acyclic, 
the complexity status of $k$-DST is not known even for 
a very strict special case that $k=2$ and $|T|=2$.

In this paper, we make a progress toward developing a non-trivial
approximation algorithm for $k$-DST.
We present an $O(D\cdot k^{D-1}\cdot\log n)$-approximation 
algorithm for $k$-DST on directed acyclic graphs (DAGs) with $D$
layers, which can be extended to a special case of $k$-DST 
on ``general graphs'' when an instance has 
a {\em $D$-shallow} optimal solution, i.e., 
there exist $k$ edge-disjoint $r,t$-paths, each of length at most $D$,
for every terminal $t\in T$.
For the case $k=1$ (DST), our algorithm yields an approximation ratio
of $O(D\log h)$, thus implying an $O(\log^3 h)$-approximation
algorithm for DST that runs in quasi-polynomial-time
(due to the height-reduction of Zelikovsky [Algorithmica'97]).
Our algorithm is based on an LP-formulation that allows us to embed a
solution to a tree-instance of GST, 
which does not preserve connectivity. 
We show, however, that one can randomly extract a solution of $k$-DST
from the tree-instance of GST.

Our algorithm is almost tight when $k, D$ are constants since the
case that $k=1$ and $D=3$ is NP-hard to approximate to within a factor
of $O(\log h)$, and our algorithm archives the same approximation 
ratio for this special case. 
We also remark that the $k^{1/4-\epsilon}$-hardness instance of
$k$-DST is a DAG with $6$ layers, and our algorithm gives 
$O(k^5\log n)$-approximation for this special case.
Consequently, as our algorithm works for general graphs,
we obtain an $O(D\cdot k^{D-1}\cdot\log n)$-approximation
algorithm for a $D$-shallow instance of 
the {\em $k$ edge-connected directed Steiner subgraph} problem, where
we wish to connect every pair of terminals by $k$ edge-disjoint paths.
\end{abstract}

\section{Introduction}
\label{sec:intro}

Network design is an important class of problems in Combinatorial
Optimization and Theoretical Computer Science as 
it formulates scenarios that appear in practical settings. 
In particular, we might wish to design an overlay network that
connects a server to clients, and this can be formulated as 
the {\em Steiner tree} problem.
In a more general setting, we might have an additional constraint
that the network must be able to function after 
link or node failures, leading to 
the formulation of the {\em survivable network design} problem. 
These problems are well-studied in symmetric case
where a network can be represented by an undirected graph. 
However, in many practical settings,
links in networks are not symmetric.
For example, we might have different upload and download bandwidths in
each connection, and sometimes, transmissions are only allowed in one
direction.
This motivates the study of network design problems in directed
graphs, in particular, {\em directed Steiner} problems.

One of the most well-known directed network design problem 
is the {\em directed Steiner tree} problem (DST), which asks to find a
minimum-cost subgraph that connects a given root vertex to each
terminal.
DST is a notorious problem as there is no known polynomial-time
algorithm that gives an approximation ratio better than polynomial. 
A polylogarithmic approximation can be obtained only when an algorithm
is allowed to run in quasi-polynomial-time
\cite{CharikarCCDGGL99,Rothvoss11,FriggstadKKLST14}.
A natural generalization of DST, namely, 
the {\em $k$ edge-connected directed Steiner tree} ($k$-DST) problem,
where we wish to connect a root vertex to each terminal by $k$
edge-disjoint paths, is even more mysterious as there is no known
non-trivial approximation algorithm, despite having been studied and
mentioned many times in literature, e.g., 
by Feldman~et~al.~\cite{FeldmanKN12}, 
by Cheriyan~et~al.~\cite{CheriyanLNV14}
and by Laekhanukit~\cite{Laekhanukit14}. 

The focus of this paper is in studying the approximability of $k$-DST.
Let us formally describe $k$-DST. 
In $k$-DST, we are given a directed graph $G$ with edge-costs
$\{c_e\}_{e\in E(G)}$, a root vertex $r$ and 
a set of terminals $T\subseteq V(G)$.
The goal is to find a min-cost subgraph $H\subseteq G$ such that
$H$ has a $k$ edge-disjoint directed $r,t$-paths from the root $r$ to
each terminal $t\in T$. 
Thus, removing any $k-1$ edges from $H$ leaves at least one path from
the root $r$ to each terminal $t\in T$,
and DST is the case when $k=1$ (i.e., we need only one path).
The complexity status of $k$-DST tends to be negative. 
It was shown by Cheriyan~et~al. \cite{CheriyanLNV14} that
the problem is at least as hard as the {\em label cover} problem.
Specifically, $k$-DST admits no
$2^{\log^{1-\epsilon}n}$-approximation, for any $\epsilon>0$,
unless $\mathrm{NP}\subseteq\mathrm{DTIME}(2^{\polylog(n)})$.
Laekhanukit~\cite{Laekhanukit14}, subsequently, showed that
$k$-DST admits no $k^{1/4-\epsilon}$-approximation 
unless $\mathrm{NP}=\mathrm{ZPP}$.
The integrality gap of a natural LP-relaxation for $k$-DST is
$\Omega(k/\log k)$ which holds even for a special case of 
{\em connectivity-augmentation} where we wish to increase 
a connectivity of a graph by one.
All the lower bound results are based on the same construction
which are directed acyclic graphs (DAGs) with diameter $5$, i.e.,
any path in an input graph has length (number of edges) at most $5$
(we may also say that it has {\em $6$ layers}).
Even for a very simple variant of $k$-DST, namely $(1,2)$-DST, 
where we have two terminals,
one terminal requires one path from the root and 
another terminal requires $2$ edge-disjoint paths, 
it was not known whether the problem is NP-hard or 
polynomial-time solvable.
To date, the only known positive result for $k$-DST is 
an $O(n^{kh})$-time (exact) algorithm  for $k$-DST on DAGs
\cite{CheriyanLNV14}, which thus runs in polynomial-time when $kh$ is  
constant,
and a folk-lore $|T|$-approximation algorithm, which can
be obtained by computing min-cost $k$-flow for $|T|$ times, 
one from the root $r$ to each terminal $t$
and then returning the union as a solution.
We emphasize that there was no known non-trivial approximation
algorithm even when an input graph is ``directed acyclic'' and 
has ``constant number of layers''.
Also, in contrast to DST, in which an $O(2^{|T|}\poly(n))$-time
(exact) algorithm exists for general graphs,
it is not known whether $k$-DST for $k=2$ and $|T|=2$ is
polynomial-time solvable if an input graph is not acyclic.
This leaves challenging questions whether ones can design a non-trivial 
approximation algorithm for $k$-DST on DAGs with at most $D$ layers,
and whether ones can design a non-trivial approximation algorithm
when an input graph is not acyclic.

In this paper, we make a progress toward developing a non-trivial
approximation algorithm for $k$-DST.
We present the first ``non-trivial'' approximation algorithm for
$k$-DST on DAGs with $D$ layers that achieves an approximation ratio
of $O(D\cdot k^{D-1}\cdot\log n)$.
Our algorithm can be extended to a special case of $k$-DST 
on ``general graphs'' where an instance has 
a {\em $D$-shallow} optimal solution, i.e., 
there exist $k$ edge-disjoint $r,t$-paths, each of length at most $D$,
for every terminal $t\in T$.
Consequently, as our algorithm works for a general graph,
we obtain an $O(D\cdot k^{D-1}\cdot\log n)$-approximation
algorithm for a $D$-shallow instance of 
the {\em $k$ edge-connected directed Steiner subgraph} problem, where  
we wish to connect every pair of terminals by $k$ edge-disjoint paths,
i.e., the set of terminal $T$ is required to be $k$-edge connected in
the solution subgraph (there is no root vertex in this problem).

Our algorithm is almost tight when $k$ and $D$ are constants
because the case that $k=1$ and $D=3$ is essentially 
the {\em set cover} problem, which is NP-hard to approximate to within
a factor of $O(\log h)$ \cite{LundY94,Feige98},
and our algorithm achieves the same approximation ratio.
We also remark that the $k^{1/4-\epsilon}$-hardness instance of
$k$-DST is a DAG with $6$ layers, and our algorithm gives 
$O(k^5\log n)$-approximation for this special case.
For $k=1$, we obtain a slightly better bound of $O(D\log h)$,
thus giving an LP-based $O(\log^3 h)$-approximation algorithm for DST
as a by product. 

The key idea of our algorithm is to formulate an LP-relaxation
with a special property that a fractional solution
can be embedded into a tree instance of 
the {\em group Steiner tree} problem (GST).
Thus, we can apply the GKR Rounding algorithm in \cite{GargKR00}
for GST on trees to round the fractional solution.
However, embedding of an LP-solution to a tree instance of GST does
not preserve connectivity.
Also, it does not lead to a reduction from $k$-DST to 
the $k$ edge-connected variant of GST, namely, $k$-GST. 
Hence, our algorithm is, although simple, not straightforward. 

\subsection{Our Results}

Our main result is an $O(D \cdot k^{D-1}\cdot \log n)$-approximation
algorithm for $k$-DST on a $D$-shallow instance,
which includes a special case that an input graph 
is directed acyclic and has at most $D$ layers. 

\begin{theorem}
\label{thm:approx-k-dst}
Consider the $k$ edge-connected directed Steiner tree problem.
Suppose an input instance has an optimal solution $H^*$ in which,
for every terminal $t\in T$, 
$H^*$ has $k$ edge-disjoint $r,t$-paths 
such that each path has length at most $D$. 
Then there exists an $O(D\cdot k^{D-1}\cdot \log n)$-approximation
algorithm.
In particular, there is an 
$O(D\cdot k^{D-1}\cdot \log n)$-approximation
algorithm for $k$-DST on a directed acyclic graph with $D$ layers.
\end{theorem}

For the case $k=1$, our algorithm yields a slightly better guarantee
of $O(D\log h)$. 
Thus, we have as by product an LP-based approximation algorithm 
for DST.
Applying Zelikovsky's height-reduction theorem
\cite{Zelikovsky97,HelvigRZ01},
this implies an LP-based quasi-polynomial-time
$O(\log^3 h)$-approximation algorithm for DST.
(The algorithm runs in time $O(\poly(n^D)$ and has approximation ratio
$O(h^{1/D}\cdot D^2 \log h)$.)

Theorem~\ref{thm:approx-k-dst} also implies an algorithm of the same 
(asymptotic) approximation ratio for a $D$-shallow instance of 
the $k$ edge-connected directed Steiner subgraph problem, where we
wish to find a subgraph $H$ such that the set of terminal $T$ is
$k$-edge-connected in $H$. 
To see this, we invoke the algorithm in
Theorem~\ref{thm:approx-k-dst} as follows.
Take any terminal $t^*\in T$ as a root vertex of a $k$-DST instance.
Then we apply the algorithm for $k$-DST to find a subgraph $H^{out}$
such that every terminal is $k$ edge-connected from $t^*$.
We apply the algorithm again to find a subgraph $H^{in}$ 
such that every terminal is $k$ edge-connected to $t^*$.
Then the set of terminal $T$ is $k$-edge connected in the graph
$H^{out}\cup H^{in}$ by transitivity of edge-connectivity.
The cost incurred by this algorithm is at most twice that of 
the algorithm in Theorem~\ref{thm:approx-k-dst}.
Thus, we have the following theorem as a corollary of 
Theorem~\ref{thm:approx-k-dst}

\begin{theorem}
\label{thm:approx-k-conn-subgraph}
\label{thm:approx-k-dst}
Consider the $k$ edge-connected directed Steiner subgraph problem.
Suppose an input instance has an optimal solution $H^*$ in which,
for every pair of terminals $s,t\in T$, 
$H^*$ has $k$ edge-disjoint $s,t$-paths 
such that each path has length at most $D$. 
Then there exists an $O(D\cdot k^{D-1}\cdot \log n)$-approximation
algorithm.
\end{theorem}

\paragraph{Overview of our algorithm}

The key idea of our algorithm is to embed an LP solution 
for $k$-DST to a standard LP of GST on a tree.
(We emphasize that we embed the LP solution of $k$-DST to that of GST
not $k$-GST.) 
At first glance, a reduction from $k$-DST to GST on trees
is unlikely to exist because any such reduction 
would destroy all the connectivity information.
We show, however, that such tree-embedding exists,
but we have to sacrifice running-time and cost to obtain such
embedding.

The reduction is indeed the same as a folk-lore reduction from DST
to GST on trees. 
That is, we list all rooted-paths (paths that start from the root
vertex) of length at most $D$ in an input graph and form a suffix
tree. 
In the case of DST, if there is an optimal solution which is a tree of
height $D$, then it gives an approximation preserving reduction from
GST to DST which blows up the size (and thus the running time) of the 
instance to $O(n^D)$.  
Unfortunately, for the case of $k$-DST with $k>1$, 
this reduction does not give an equivalent reduction 
from $k$-DST to $k$-GST on trees. 
The reduction is valid in one direction, i.e., any feasible solution
to $k$-DST has a corresponding feasible solution to the tree-instance
of $k$-GST. 
However, the converse is not true as a feasible solution to the
tree-instance of $k$-GST might not give a feasible solution to $k$-DST.
Thus, our reduction is indeed an ``invalid'' reduction from $k$-DST to
a tree instance of ``GST'' (the case $k=1$).

To circumvent this problem, we formulate an LP that provides a
connection between an LP solution on an input $k$-DST instance 
and an LP solution of a tree-instance of GST. 
Thus, we can embed an LP solution to an LP-solution of GST on a 
(very large) tree.
We then round the LP solution using the GKR Rounding algorithm for GST on
trees \cite{GargKR00}.
This algorithm, again, does not give a feasible solution to $k$-DST as
each integral solution we obtain only has ``connectivity one'' 
and thus is only feasible to DST.
We cope with this issue by using a technique 
developed by Chalermsook~et~al. in \cite{ChalermsookGL15}.
Specifically, we sample a sufficiently large number of
DST solutions and show that the union of all these solutions
is feasible to $k$-DST using cut-arguments.

Each step of our algorithm and the proofs are mostly standard,
but ones need to be careful in combining each step.
Otherwise, the resulting graph would not be feasible to $k$-DST.

\medskip\noindent{\bf Organization.}
We provide definitions and notations in Section~\ref{sec:prelim}. 
We start our discussion by presenting a reduction from DST to GST
in Section~\ref{sec:DST-to-GST}.
Then we discuss properties of minimal solutions for $k$-DST
in Section~\ref{sec:props-min-solution}.
We present standard LPs for $k$-DST and GST in
Section~\ref{sec:standard-LPs}
and formulate a stronger LP-relaxation for $k$-DST in
Section~\ref{sec:strong-LP-k-DST}.
Then we proceed to present our algorithm in
Section~\ref{sec:algo-kDST}. 
Finally, we provide some discussions in Section~\ref{sec:conclusion}.

\section{Preliminaries}
\label{sec:prelim}

We use standard graph terminologies. 
We refer to a directed edge $(u,v)$, shortly, by $uv$
(i.e., $u$ and $v$ are head and tail of $uv$, respectively),
and we refer to an undirected edge by $\{u,v\}$.
For a (directed or undirected) graph $G$, we denote by $V(G)$ and $E(G)$
the sets of vertices and edges of $G$, respectively.
If a graph $G$ is associated with edge-costs $\{c_e\}_{e\in E(G)}$, 
then we denote the cost of any subgraph $H\subseteq G$ by
$\cost(H) = \sum_{e\in E(H)}c_e$.
For any path $P$, 
we use {\em length} to mean the number of edges in a path $P$ and
use {\em cost} to mean the total costs of edges in $P$.
%

In the {\em directed Steiner tree} problem (DST), 
we are given a directed graph $G$ with edge-costs $\{c_e\}_{e\in E(G)}$, 
a root vertex $r$ and a set of terminals $T\subseteq V(G)$.
The goal is to find a min-cost subgraph $H\subseteq G$ such that
$H$ has a directed path from the root $r$ to each terminal $t\in T$.
A generalization of DST is the 
{\em $k$ edge-connected directed Steiner tree} problem ($k$-DST).
In $k$-DST, we are given the same input as in DST plus an integer $k$. 
The goal is to find a min-cost subgraph $H$ that has $k$ edge-disjoint
paths from the root $r$ to each terminal $t\in T$.
The {\em $k$ edge-connected directed Steiner subgraph} problem
is a variant of $k$-DST, where there is no root vertex, and
the goal is to find a min-cost subgraph $H$ such that
the set of terminals $T$ is $k$ edge-connected in $H$. 

The problems on undirected graphs that are closely related to 
of DST and $k$-DST are the 
{\em group Steiner tree} problem (GST) and 
the {\em $k$ edge-connected group Steiner tree} problem ($k$-GST).
In GST, we are given an undirected graph $\calG$ with edge-costs
$\{c_e\}_{e\in E(G)}$, a root vertex $r$ and a collection of subset of
vertices $\{\calT_i\}_{i=1}^h$ called groups. The goal is to find a
a min-cost subgraph $\calH$ that connects $r$ to each group
$\calT_i$. In $k$-GST, the input consists of an additional integer $k$,
and the goal is to find a min-cost subgraph $H$ with $k$ edge-disjoint
$r,\calT_i$-paths for every group $T_i$.

Consider an instance of DST (resp., $k$-DST).
We denote by $Q$ the set of all paths in $G$ that start from the root $r$.
The set of paths in $Q$ that end with a particular pattern, say
$\sigma=(v_1,\ldots,v_q)$, is denoted by $Q(\sigma)$. This pattern $\sigma$
can be a vertex $v$, an edge $e$ or a path $\sigma=(v_1,\ldots,v_q)$ in
$G$. For example, $Q(u,v,w)$ consists of paths $P$ of the form
$P=(r,\ldots,u,v,w)$. We say that a path $P$ ends at a vertex $v$
(resp., an edge $e$) if $v$ (resp., $e$) is 
the last vertex (resp., edge) of $P$.

We may consider only paths with particular length, say $D$.
We denote by $Q_D$ the set of paths that start at $r$ and has length
at most $D$. The notation for $Q_D$ is analogous to $Q$, e.g.,
$Q_D(uv)\subseteq Q_D$ is the set of paths in $Q_D$ that end at an
edge $uv$.
A concatenation of a path $p$ with an edge $e$ or a vertex $v$
are denoted by $p+e$ and $p+v$, respectively.
For example, $(u_1,\ldots,u_{\ell})+vw=(u_1,\ldots,u_{\ell},v,w)$.

Given a subset of vertices $S$, 
the set of edges entering $S$ is denoted by 
\[
\delta^-(S) = \{uv \in E: u\in S, v\not\in S\}
\]
The indegree of $S$ is denoted by $\indeg(S) = |\delta^-(S)|$.
Analogously, we use $\delta^+(S)$ and $\outdeg(S)$ for
the set of edges leaving $S$.
For undirected graphs, we simply use the notations $\delta(S)$ and $\deg(S)$.

We say that a feasible solution $H$ to $k$-DST is {\em $D$-shallow}
if, for every terminal $t\in T$, there exists a set of 
$k$ edge-disjoint $r,t$-paths in $H$ such that every path has length at most
$D$.
An instance of $k$-DST that has an optimal $D$-shallow solution
is called a $D$-shallow instance.
We also use the term $D$-shallow analogously for $k$-GST
and the $k$ edge-connected Steiner subgraph problem.
%


To distinguish between the input of $k$-DST (which is a directed graph)
and $k$-GST (which is an undirected graph), we use script fonts, 
e.g., $\calG$, to denote the input of $k$-GST.
Also, we use $\calQ$ to denote the set of all paths from the root $r$
to any vertex $v$ in the graph $\calG$. 
The cost of a set of edges $F$ (or a graph) is defined by a function 
$\cost(F)=\sum_{e\in F}c_e$. 
At each point, we consider only one instance of $k$-DST 
(respectively, $k$-GST).
So, we denote the cost of the optimal solution to $k$-DST by 
$\opt_{kDST}$ (respectively, $\opt_{kGST}$).

\section{Reduction from Directed Steiner Tree to Group Steiner Tree}
\label{sec:DST-to-GST}

In this section, we describe a reduction $\calR$ from DST to GST.
We recall that $Q$ denotes all the $r,v$-paths in a DST instance $G$. 
The reduction is by simply listing paths in the directed graph $G$
as vertices in a tree $\calG=\calR(G)$ and
joining each path $p$ to $p+e$ if $p+e$ is a path in $G$. 
In fact, $\calR(G)$ is a {\em suffix tree} of paths in $Q$.
To be precise,
\begin{align*}
V(\calG) &= \{p : \mbox{$p$ is an $r,v$-path in $G$}\}\\
E(\calG) &= \{\{p,p+e\}: \mbox{both $p$ and $p+e$ are paths in $G$ starting at $r$}\}
\end{align*}

We set the cost of edges of $\calG$ to be $c_{\{p,p+e\}}=c_e$.
Since the root $r$ has no incoming edges in $G$,
$r$ maps to a unique vertex $(r) \in \calG$, and we define $(r)$ as
the root vertex of the GST instance.
We will abuse $r$ to mean both the root of DST and 
its corresponding vertex of GST. 
For each terminal $t_i\in T$, define a group of the GST instance as
\[
\calT_i := Q(t_i) = \{p \subseteq G: \mbox{$p$ is an $r,t_i$-path in $G$}\}
\]

It is easy to see that the reduction $\calR$ produces a tree,
and there is a one-to-one mapping between 
a path in the tree $\calG=\calR(G)$ and a path in
the original graph $G$.
Thus, any tree in $G$ corresponds to a subtree of $\calR(G)$ 
(but not vice versa),
which implies that the reduction $\calR$ is approximation-preserving 
(i.e., $\opt_{DST}=\opt_{GST}$). 
Note, however, that the size of the instance blows up 
from $O(n+m)$ to $O(n^D)$, 
where $D$ is the length of the longest path in $G$.
The reduction holds for general graphs, but
it is approximation-preserving only if the DST instance
is $D$-shallow, i.e., it has an optimal solution $H^*$ 
such that any $r,t_i$-path in $H^*$ has length at most $D$, 
for all terminals $t_i\in T$.
However, Zelikovsky~\cite{Zelikovsky97,HelvigRZ01} showed
that the {\em metric completion} of $G$ always contains a $D$-shallow 
solution with cost at most $D|V(G)|^{1/D}$ of 
an optimal solution to DST. 
(This is now known as Zelikovsky's height reduction theorem.)
Thus, we may list only paths of length at most $D$ 
from the metric completion.
We denote the reduction that lists only paths of length at most $D$ by
$\calR_D$.



\section{Properties of Minimal Solutions to $k$-DST}
\label{sec:props-min-solution}

In this section, we provide structural lemmas which are building
blocks in formulating a strong LP-relaxation for $k$-DST. 
These lemmas characterize properties of 
a minimal solution to $k$-DST.

\begin{lemma}
\label{lem:rv-paths-in-minimal-kDST}
Let $H$ be any minimal solution to $k$-DST.
Then $H$ has at most $k$ edge-disjoint $r,v$-paths,
for any vertex $v\in V(H)$.
\end{lemma}
\begin{proof}
Suppose to a contrary that $H$ has $k+1$ edge-disjoint $r,v$-paths,
for some vertex $v\in V(H)$.
Then $v$ must have indegree at least $k+1$ in $H$.
We take one of the $k-1$ edges entering $v$, namely, $uv$.
By minimality of $H$, removing $uv$ results in a graph 
$H'=H\setminus uv$ that has less than 
$k$ edge-disjoint $r,t_i$-paths for some terminal $t_i\in T$.
Thus, by Menger's theorem, 
there must be a subset of vertices $S\subseteq V$
such that $t_i\in S$, $r\in V\setminus S$ and 
$\indeg_{H'}(S) \leq k-1$. 
Observe that we must have $uv$ in $\delta_H^-(S)$ because
$H$ is a feasible solution to $k$-DST, 
which means that $v\in S$.
Since we remove only one edge $uv$ from $H$, 
the graph $H'$ must have $k$ edge-disjoint $r,v$-paths. 
But, this implies that $\indeg_{H'}(S) \geq k$,
a contradiction.
\end{proof}

\begin{lemma}
\label{lem:paths-eq-indeg}
Let $H$ be any minimal solution to $k$-DST.
Any vertex $v\in V(H)$ has indegree exactly $\lambda(v)$, 
where $\lambda(v)$ is the maximum number of 
edge-disjoint $r,v$-paths in $H$.
\end{lemma}

\begin{proof}
The proof follows a standard uncrossing argument.
Assume a contradiction that 
$v$ has indegree at least $\lambda(v)+1$ in $H$.
By Menger's theorem, there is a subset of 
vertices $U\subseteq V$ such that 
$\indeg_H(U) = \lambda(v)$, $v\in U$ and $r\not\in U$ 
that separates $v$ from $r$.
We assume that $U$ is a minimum such set. 
Since $\indeg_H(v) > \lambda(v)$, there is an edge
$uv\in E(H)$ that is not contained in $\delta^-_H(U)$,
i.e., $u,v\in U$. 

By minimality of $H$, removing $uv$ results in 
the graph $H'=H\setminus uv$ such that $H'$ has less than
$k$ edge-disjoint $r,t_i$-path for some terminal $t_i\in T$.
Thus, by Menger's theorem, 
there is a subset of vertices $W$
such that $t_i\in W$, $r\not\in W$, $uv\in\delta_H^-(W)$ 
and $\indeg_H(W)=k$.
(The latter is because $H$ is a feasible solution to $k$-DST.)

Now we apply an uncrossing argument to $U$ and $W$.
By submodularity of $\indeg_H$, we have
\[
\indeg_H(U) + \indeg_H(W) 
  \geq \deg_H(U\cap W) + \deg_H(U\cup W)
\]
Observe that $v\in U\cap W$, $t\in U\cup W$
and $r\not\in S\cup S'$. 
So, by the edge-connectivity of $v$ and $t$,  
\begin{equation}
\label{eq:submod}
\indeg_H(U\cap W) \geq \lambda(v) 
  \quad \mbox{and} \quad
\indeg_H(U\cup W) \geq k
\end{equation}
The sum of the left-hand side of Eq~\eqref{eq:submod} is 
\[
\indeg_H(U) + \indeg_H(W)  = k + \lambda(v).
\]
So, we conclude that 
\[
\indeg_H(U\cap W) = \lambda(v)
  \quad \mbox{and} \quad
\indeg_H(U\cup W) = k
\]
Consequently, we have the set $U' = U\cap W$ such that 
$\indeg_H(U')=\lambda(v)$, $v\in U'$ and $r\not\in U'$ that
separates $v$ from $r$.
Since $u\not\in W$, we know that $U'$ is strictly smaller than $U$.
This contradicts to the minimality of $U$.
\end{proof}

The following is a corollary of Lemma~\ref{lem:rv-paths-in-minimal-kDST}
and Lemma~\ref{lem:paths-eq-indeg}
\begin{corollary}
\label{cor:kdst-maxdeg-k}
Let $H$ be a minimal solution to $k$-DST.
Then any vertex $v\in V(H)$ has indegree at most $k$.
\end{corollary}

The next lemma follows from Corollary~\ref{cor:kdst-maxdeg-k}.

\begin{lemma}
\label{lem:no-of-paths-k-DST}
Consider any minimal solution $H$ to $k$-DST
(which is a simple graph).
For any edge $e\in E(H)$ and $\ell\geq 2$, 
there are at most $k^{\ell-2}$ paths in $H$ with length at most $\ell$
that start at the root $r$ and ends at $e$.
That is, 
$|Q^H_{\ell}(e)| \leq k^{\ell-2}$ for all $e\in E(H)$,
where $Q^H_{\ell}(e)$ is 
the set of $r,v$-paths of length $\ell$ in $H$. 
\end{lemma}
\begin{proof}
We prove by induction.
The base case $\ell=2$ is trivial because any rooted path of
length at most $2$ cannot have a common edge.

Assume, inductively, that $|Q^H_{\ell-1}(e)| \leq k^{\ell-3}$ 
for some $\ell \geq 3$. 
Consider any edge $vw\in E(H)$.
By Corollary~\ref{cor:kdst-maxdeg-k}, $v$ has indegree at most $k$.
Thus, there are at most $k$ edges entering $v$, namely,
$u_1v,\ldots,u_dv$, where $d=\indeg(v)$. 
By the induction hypothesis, each edge is the last edge of at most
$k^{\ell-3}$ paths in $Q^H_{\ell-1}$. 
Thus, we have at most $d \cdot k^{\ell-3} \leq k^{\ell-2}$ paths that
end at $uv$. That is,
\[
|Q^H_{\ell}(vw)| 
  \leq \sum_{i=1}^d|Q^H_{\ell-1}(u_dv)|
  \leq \sum_{i=1}^dk^{\ell-3}
  = d \cdot k^{\ell-3}
  \leq k^{\ell-2}.
\]
\end{proof}


\section{Standard LPs for $k$-DST and GST}
\label{sec:standard-LPs}

In this section, we describe standard LPs for $k$-DST and GST. 
Each LP consists of two sets of variables, 
a variable $x_e$ on each edge $e$ and 
a variable $f^i_p$ on each path $p$ and a terminal $t_i$.
The variable $x_e$ indicates whether we choose an edge $e$
in a solution.
The variable $f^i_p$ is a flow-variable on each path and 
thus can be written in a compact form 
using a standard flow formulation.

\[
\mbox{LP-k-DST}\left\{
\begin{array}{lrll}
  \min & \sum_{e\in E(G)}c_ex_e\\
  \mbox{s.t.}
    & \sum_{p\in Q(t_i):e\in E(p)}f^i_p  &\leq x_e  & \forall e\in E(G), 
                                     \forall t_i\in T\\
    & \sum_{p\in Q(t_i)}f^i_p &\geq k & \forall t_i\in T\\
    & x_e                &\leq 1   & \forall e\in E(G)\\
    & x_e                &\geq 0   & \forall e\in E(G)\\
    & f^i_p              &\geq 0   & \forall p\in Q(t_i), 
                                     \forall t_i\in T.
\end{array}
\right. 
\]

The standard LP for GST is similar to LP-k-DST.

\[
\mbox{LP-GST}\left\{
\begin{array}{lrll}
  \min & \sum_{e\in E(\calG)}c_ex_e\\
  \mbox{s.t.}
    & \sum_{v\in \calT_i}\sum_{p\in \calQ(v):e\in E(p)}f^i_p &\leq x_e  
                                   & \forall e\in E(\calG), 
                                     \forall i=1,2,\ldots,h\\
    & \sum_{v\in \calT_i}\sum_{p\in \calQ(v)}f^i_p &\geq 1    
                                   & \forall i=1,2,\ldots,h\\
    & x_e                &\leq 1   & \forall e\in E(\calG)\\
    & x_e                &\geq 0   & \forall e\in E(\calG)\\
    & f^i_p              &\geq 0   & \forall p\in \calQ,
                                     \forall i=1,2,\ldots,h
\end{array}
\right. 
\]

\section{A Strong LP-relaxation for for $k$-DST}
\label{sec:strong-LP-k-DST}

In this section, we formulate a strong LP-relaxation for $k$-DST that
allows us to embed a fractional solution into an LP solution of LP-GST
on a tree. 

\[
\mbox{LP-k-DST*}\left\{
\begin{array}{lrllr}
  \min & \sum_{e\in E}c_ex_e\\
  \mbox{s.t.}
    & \sum_{p\in Q(t_i):e\in E(p)}f^i_p  &\leq x_e  & \forall e\in E(G), 
                                     \forall t_i\in T\\
    & \sum_{p\in Q(t_i)}f^i_p &\geq k & \forall t_i\in T\\
    & \sum_{p\in Q(t_i):q\subseteq p}f^i_p  &\leq y_q  & \forall q\in Q, 
                                     \forall t_i\in T
                        &\mbox{(Subflow Capacity)}\\
    & \sum_{p\in Q_{\ell}(e)}y_p &\leq \max\{1,k^{\ell-2}\} \cdot x_e
                        & \forall e\in E, \forall \ell \geq 1
                        & \mbox{(Aggregating $k$-Flow)}\\
    & x_e                &\leq 1   & \forall e\in E(G)\\
    & x_e                &\geq 0   & \forall e\in E(G)\\
    & f^i_p              &\geq 0   & \forall p\in Q(t_i), 
                                     \forall t_i\in T\\
    & y_p                &\geq 0   & \forall p\in Q 
\end{array}
\right. 
\]

For $D$-shallow instances of $k$-DST, we replace $Q$ by $Q_D$
to restrict length of paths to be at most $D$. 
The next lemma shows that LP-k-DST* is an LP-relaxation for $k$-DST.

\begin{lemma} \label{lem:valid-LP-k-DST*}
LP-k-DST* is an LP-relaxation for $k$-DST.
Moreover, replacing $Q$ by $Q_D$ gives an LP-relaxation for $k$-DST on
$D$-shallow instances.
\end{lemma}

\begin{proof}
LP-k-DST* is, in fact, obtained from LP-k-DST (which is a standard LP)
by adding a new variable $y_p$ and two constraints.
\begin{itemize}
\item[(1)] {\bf Subflow-Capacity:}
  $\displaystyle\sum_{p\in Q(t_i):q\subseteq p}f^i_p \leq y_q,
     \forall q\in Q, \forall t_i\in T$.
\item[(2)]{\bf Aggregating $k$-Flow:}
  $\displaystyle\sum_{p\in Q_{\ell}(e)}y_p 
     \leq \max\{1,k^{\ell-2}\}\cdot x_e,  
     \forall e\in E, \forall \ell \geq 1$.
\end{itemize}

To show that these two constraints are valid for $k$-DST,
we take a minimal feasible ($D$-shallow) solution
$H$ of $k$-DST.
We define a solution $(x,f,y)$ to LP-k-DST as below.
\[
\begin{array}{lcl}
  \begin{array}{rl}
    x_e &= \left\{\begin{array}{rl}
              1 & \mbox{if $e\in E(H)$}\\
              0 & \mbox{otherwise}
           \end{array}\right.
  \end{array}
 & \qquad &
  \begin{array}{rl}
    y_p &= \left\{\begin{array}{rl}
              1 & \mbox{if $p\subseteq H \land p\in Q$}\\
              0 & \mbox{otherwise}
           \end{array}\right.
  \end{array}\\
  \begin{array}{rl}
    f^i_p &= \left\{\begin{array}{rl}
              1 & \mbox{if $p\subseteq H \land p\in Q(t_i)$}\\
              0 & \mbox{otherwise}
           \end{array}\right.
  \end{array}\\
\end{array}
\]
By construction, $f^i_p = 1$ implies that $y_p=1$.
Thus, $(x,f,y)$ satisfies the Subflow-Capacity constraint.
By minimality of $H$, Corollary~\ref{lem:no-of-paths-k-DST} 
implies that even if we list all the paths of length $\ell\geq 2$ in
$H$, at most $k^{\ell-2}$ of them end at the same edge, and we know
that rooted paths of length one share no edge 
(given that $H$ is a simple graph).
Thus, $(x,f,y)$ satisfies the Aggregating $k$-Flow constraint.
Consequently, these two constraints are valid for $k$-DST.

On the other hand, any integral solution that is not feasible to
$k$-DST could not satisfy the constraints of LP-k-DST* simply
because LP-k-DST* contains the constraints of LP-k-DST, which is a
standard LP for $k$-DST.
Thus, LP-k-DST* is an LP-relaxation for $k$-DST.

The proof for the case of $D$-shallow instances is the same as above  
except that we take $H$ as a minimal $D$-shallow solution
and replace $Q$ by $Q_D$. 
\end{proof}

\section{An Approximation Algorithm for $k$-DST}
\label{sec:algo-kDST}

In this section, we present an approximation algorithm
for $k$-DST on a $D$-shallow instance. 
Our algorithm is simple.
We solve LP-k-DST* on an input graph $G$ 
and then embed an optimal fractional solution $(x,f,y)$ to 
an LP-solution $(\hat{x},\hat{f})$ of LP-GST on the tree $\calR(G)$. 
We lose a factor of $O(k^{D-2})$ in this process.
As we now have a tree-embedding of an LP-solution,
we can invoke the GKR Rounding algorithm \cite{GargKR00} to round
an LP-solution on the tree $\calR(G)$.
Our embedding guarantees that any edge-set of size $k-1$ 
in the original graph $G$ never maps to 
an edge-set in the tree $\calG=\calR(G)$ that 
separates $r$ and $\calT_i=Q(t_i)$ in $\calG$. 
So, the rounding algorithm still outputs a feasible solution to GST
with constant probability even if we remove edges in the tree
$\calG$ that correspond to a subset of $k-1$ edges in $G$.
Consequently, we only need to run the algorithm for 
$O(D\cdot k\cdot \log n)$ times
to boost the probability of success so that,
for any subset of $k-1$ edges and any terminal $t_i\in T$,
we have at least one solution that 
contains an $r,t_i$-path using none of these $k-1$ edges.
In other words, the union of all the solutions
satisfies the connectivity requirement.
Our algorithm is described in Algorithm~\ref{algo:kdst}.

\begin{algorithm}
\caption{Algorithm for $k$-DST}
\begin{algorithmic}
\label{algo:kdst}
\STATE Solve LP-k-DST* and obtain an optimal solution $(x,f,y)$.
\STATE Map $(x,f,y)$ to a solution $(\hat{x},\hat{f})$ to LP-GST on $\calG=\calR(G)$.
\FOR{$i=1$ \TO $2Dk \log_2 n$}
   \STATE Run GKR Rounding on $(\hat{x},\hat{f})$ to get a solution $\calZ_i$. 
   \STATE Map $\calZ_i$ back to a subgraph $Z_i$ of $G$.
\ENDFOR
\RETURN $H=\bigcup_iZ_i$ as a solution to $k$-DST.
\end{algorithmic}
\end{algorithm}

We map a solution $(x,f,y)$ of LP-k-DST* on $G$ to 
a solution $(\hat{x},\hat{f})$ of LP-GST on the tree $\calG=\calR(G)$
as below.
Note that there is a one-to-one mapping between a path in $G$ and 
a path in the tree $\calG$.
\[
\begin{array}{rll}
\hat{x}_{\{p,p+e\}}  &:= y_{p+e} & \mbox{ for all $p+e \in Q$}\\
\hat{f}^i_{p}      &:= f^i_p & 
   \mbox{ for all $p \in Q$ and for all $t_i\in T$}
\end{array}
\]

\subsection{Cost Analysis}
\label{sec:cost-analysis}

We show that $\cost(\hat{x},\hat{f}) \leq k^{D-2}\cdot \cost(x,f,y)$.

\begin{lemma}
\label{lem:cost-kdst-mapping}
Consider a solution $(\hat{x},\hat{f})$ to LP-GST, which is mapped from 
a solution $(x,f)$ of LP-k-DST* when 
an input $k$-DST instance is $D$-shallow,
for $D\geq 2$.
The cost of $(\hat{x},\hat{f})$ is at most
$\cost(\hat{x},\hat{f}) \leq k^{D-2}\cdot \cost(x,f)$.
\end{lemma}
\begin{proof}
By the constraint 
$\sum_{p\in Q_{\ell}(e)}f_p \leq \max\{1,k^{\ell-2}\}\cdot x_e$,
we have that
\[
\begin{array}{rlll}
\cost(\hat{x},\hat{f}) 
  &= \sum_{e'\in E(\calG)}c_{e'}\hat{x}_{e'}
  &= \sum_{e\in E(G)}\;\sum_{\{p,p+e\}\in E(\calG)}c_e\hat{x}_{\{p,p+e\}}\\
  &= \sum_{e\in E(G)}\;\sum_{p+e\in Q}c_e f_{p+e}
  &= \sum_{e\in E(G)}\;\sum_{p\in Q(e)}c_e f_p\\
  &= \sum_{e\in E(G)}\left(c_e\cdot\sum_{p\in Q(e)}f_p\right)
  &\leq \sum_{e\in E(G)}c_e\cdot k^{D-2}\cdot x_e\\
  &= k^{D-2}\cdot\cost(x,f).
\end{array}
\]
\end{proof}

It can be seen from Algorithm~\ref{algo:kdst} and 
Lemma~\ref{lem:cost-kdst-mapping} that 
the algorithm outputs a solution $H$ with cost 
at most $O(D k^{D-1}\log n)\cdot\cost(x,f)$.
Thus, $H$ is an $O(D k^{D-1}\log n)$-approximate solution.
It remains to show that $H$ is feasible to $k$-DST.

\subsection{Feasibility Analysis}
\label{sec:feasibility-analysis}

Now we show that $H$ is feasible to $k$-DST 
with at least constant probability.
To be formal, consider any subset $F\subseteq E(G)$ of $k-1$ edges.
We map $F$ to their corresponding edges $\calF$ in the tree $\calG$.
Thus, $\calF := \{\{P,P+e\}: P+e\in Q \land e\in F\}$.

Observe that no vertices in $\calG\setminus\calF$ correspond
to paths that contain an edge in $F$.
Thus, we can define an LP solution $(y^F,z^F)$ for LP-GST
on the graph $\calG \setminus \calF$ as follows.
\[
\begin{array}{rlcrl}
y^F_e &= \left\{\begin{array}{ll}
           \hat{x}_e & \mbox{if $e\not\in\calF$}\\
           0    & \mbox{otherwise}
        \end{array}\right.
& \qquad &
z^{F,i}_p &= \left\{\begin{array}{ll}
           \hat{f}^i_p & \mbox{if $E(p)\cap\calF=\emptyset$}\\
           0    & \mbox{otherwise}
        \end{array}\right.
\end{array}
\]

We show that $(y^F,z^F)$ is feasible to LP-GST on $\calG\setminus\calF$.
\begin{lemma}
\label{lem:feasibility-of-calG-minus-F}
For any subset of edges $F\subseteq E(G)$, 
define $(y^F,z^F)$ from $(\hat{x},\hat{f})$ as above. 
Then $(y^F,z^F)$ is feasible to LP-GST on $\calG\setminus\calF$.
\end{lemma}

\begin{proof}
First, observe that $z^{F,i}_p>0$ only if a path $p$ contains no edges in $\calF$.
So, by construction, $(y^F,z^F)$ satisfies
$z^{F,i}_p = \hat{f}^i_p \leq \hat{x}_e = y^F_e$
for all $e\in E(p)$.
Hence, $(y^F,z^F)$ satisfies the capacity constraint.

Next we show that $(y^F,z^F)$ satisfies the connectivity constraint.
Consider the solution $(x,f,y)$ to LP-k-DST*.
By the feasibility of $(x,f,y)$ and the Max-Flow-Min-Cut theorem, 
the graph $G\setminus F$ with capacities $\{x_e\}_{e\in G\setminus F}$
can support a flow of value one from $r$ to any terminal $t_i$. 
This implies that
$\sum_{p\in Q(t_i):E(p)\cap F=\emptyset}f^i_p \geq 1$.
Consequently, we have 
\[
\begin{array}{rll}
\sum_{p\in Q(t_i):E(p)\cap F=\emptyset}f^i_p 
&= \sum_{p\in \calT_i:E(p)\cap F=\emptyset}\;\sum_{p'\in \calQ(v)}\hat{f}^i_{p'}\\
&= \sum_{v\in \calT_i}\;\sum_{p'\in \calQ(v):E(p')\cap \calF=\emptyset}\hat{f}^i_{p'}\\
&= \sum_{v\in \calT_i}\;\sum_{p'\in \calQ(v):E(p')\cap \calF=\emptyset}z^{F,i}_{p'}\\
&= \sum_{v\in \calT_i}\;\sum_{p'\in \calQ(v)}z^{F,i}_{p'}\\
&\geq 1.
\end{array}
\]

All the other constraints are satisfied 
because $(y^F,z^F)$ is constructed from $(\hat{x},\hat{f})$.
Thus, $(y^F,z^F)$ is feasible to LP-GST on $\calG \setminus \calF$.
\end{proof}

Lemma~\ref{lem:feasibility-of-calG-minus-F} implies that
we can run the GKR Rounding algorithm on $(y^F,z^F)$.
The following is the property of GKR Rounding.

\begin{lemma}[\cite{GargKR00}]
\label{lem:prop-of-GKR}
There exists a randomized algorithm such that,
given a solution $(\hat{x},\hat{f})$ to LP-GST on a tree $\calG$
with height $D$,
the algorithm outputs a subgraph $\calH\subseteq\calG$
so that the probability that any subset of vertices 
$U\subseteq V(\calG)$ is connected to the root is at least
\[
\Pr[\mbox{$\calH$ has an $r,U$-path.}] 
\geq \frac{\sum_{v\in U}\sum_{p\in\calQ(v)}\hat{f}^i_p}{O(D)}
\]
Moreover, the probability that each edge is chosen is at most $\hat{x}_e$.
That is, $\E[\cost(\calH)] = \cost(\hat{x},\hat{f})$.
The running time of the algorithm is $O(|E(\calG)| + |V(\calG)|)$. 
\end{lemma}
 
Since $(y^F,z^F) \leq (\hat{x},\hat{f})$ (coordinate-wise),
we can show that running GKR Rounding 
on $(\hat{x},\hat{f})$ simulates 
the runs on $(y^F,z^F)$ for all $F\subseteq E(G)$
with $|F|\leq k-1$, simultaneously.

\begin{lemma}
\label{lem:backward-feasible-kDST}
Let $\calH$ be a subgraph of $\calG$ obtained by
running GKR Rounding on $(\hat{x},\hat{f})$,
and let $H$ be a subgraph of $G$ corresponding to $\calH$.
Then, for any subset of edges $F\subseteq E(G)$ with $|F|\leq k-1$
and for any terminal $t_i\in T$, 
\[
\Pr[\mbox{$H\setminus F$ has an $r,t_i$-path}] \geq \frac{1}{O(D)}.
\]
\end{lemma}
\begin{proof}
Let us briefly describe the work of GKR Rounding.
The algorithm marks each edge $e$ in the tree with
probability $x_e/x_{\varrho(e)}$, where $\varrho(e)$ is the parent of
an edge $e$ in $\calG$, which is unique.
Then the algorithm picks an edge $e$ if all of its ancestors
are marked. 
Clearly, removing any set of edges $\calF$ from $\calG$
only affects paths that contain an edge in $\calF$.

Let $(y^F,z^F)$ be defined from $(\hat{x},\hat{f})$ as above.
This LP solution is defined on a graph $\calG\setminus\calF$.
Thus, the probability of success is not affected by removing $\calF$
from the graph. 
By Lemma~\ref{lem:feasibility-of-calG-minus-F},
we can run GKR Rounding on $(y^F,z^F)$ and obtain
a subgraph $\calH^F$ of $\calG \setminus \calF$.
Since $z^F_p \leq \hat{f}_p$ for all paths $p\in\calQ$
and $z^F_p=0$ for all $p\in \calQ:E(p)\cap \calF\neq\emptyset$,
we have from Lemma~\ref{lem:prop-of-GKR}
and Lemma~\ref{lem:feasibility-of-calG-minus-F} that
\[
\begin{array}{rl}
\Pr[\mbox{$H\setminus F$ has an $r,t_i$-path}] 
&= 
  \Pr[\mbox{$\calH\setminus\calF$ has an $r,\calT_i$-path}]\\
&\geq 
  \Pr[\mbox{$\calH^F$ has an $r,\calT_i$-path}]\\
&\geq 
  \frac{\sum_{v\in\calT_i}\;\sum_{p\in\calQ(v)}z^F_p}{O(D)}\\
&\geq 
  \frac{1}{O(D)}.
\end{array}
\]
\end{proof}

Finally, we recall that Algorithm~\ref{algo:kdst} employs
GKR Rounding on $(\hat{x},\hat{f})$ for $2Dk\log_2 n$ times.
So, for any subset of $k-1$ edges $F\subseteq E(G)$
and for any terminal $t_i\in T$,
there exists one subgraph that has an $r,t_i$-path 
that contains no edge in $F$ with large probability.
In particular, the union is a feasible solution to $k$-DST
with at least constant probability.
\begin{lemma}
Consider the run of Algorithm~\ref{algo:kdst}.
The solution subgraph $H=\bigcup_iZ_i$ 
is a feasible solution to $k$-DST
with probability at least $1/n$.
\end{lemma}
\begin{proof}
For $i=1,2,\ldots,2Dk\log_2 n$, let $Z_i$ be a subgraph of $G$ 
obtained by running GKR Rounding on $(\hat{x},\hat{f})$
and mapping the solution back to a subgraph of $G$
as in Algorithm~\ref{algo:kdst}.
By Lemma~\ref{lem:backward-feasible-kDST}, $Z_i\setminus F$ 
has an $r,t_i$-path with probability $\Omega(1/D)$.
Since each $Z_i$ is sampled independently, we have
\[
\Pr[\forall i\,\mbox{$Z_i\setminus F$ has no $r,t_i$-path}]
\leq \left(1-\frac{1}{O(D)}\right)^{2Dk\log_2 n}
\leq \left(\frac{1}{e}\right)^{2k\log_2 n}
\leq n^{-2k}.
\]
We have at most $|E(G)|^{k-1}\leq n^{2(k-1)}$ such sets $F$ and
at most $|T|\leq n$ terminals.
So, there are at most $n^{2k-1}$ bad events where 
there exists an edge-set of size $k-1$ that 
separates the root $r$ and some terminal $t_i\in T$.
Therefore, by union bound, $H=\bigcup_iZ_i$ is a feasible solution 
to $k$-DST with probability at least $1/n$.
\end{proof}

This completes the proof of Theorem~\ref{thm:approx-k-dst}.
Note that, for the case of DST ($k=1$), we only need to run GKR
Rounding for $O(D\log h)$ times, thus implying an approximation ratio of
$O(D\log h)$. 

\section{Conclusion and Discussion}
\label{sec:conclusion}

We presented the first non-trivial approximation algorithm for $k$-DST
in a special case of a $D$-shallow instance, which exploits
the reduction from DST to GST. 
We hope that our techniques will shed some light in designing 
an approximation algorithm for $k$-DST in general case and
perhaps lead to a bi-criteria approximation algorithm in the same manner
as in \cite{ChalermsookGL15}.

One obstruction in designing an approximation algorithm in directed
graphs is that there is no ``true'' (perhaps, probabilistic) tree-embedding for
directed graphs.
Both devising a tree-embedding for directed graphs and designing an
approximation algorithm for $k$-DST with $k \geq 2$ are big open
problems in the area.
Another open problem, which is considered as the most challenging one
by many experts, is whether there exists a polynomial-time algorithm
for DST that yields a sub-polynomial approximation ratio.

\subparagraph*{Acknowledgements.}
Our work was inspired by the works of 
Rothvo{\ss}~\cite{Rothvoss11} and 
Friggstad~et~al.~\cite{FriggstadKKLST14}
and by discussions with Joseph Cheriyan and Lap Chi Lau.
We also thank Zachary Friggstad for useful discussions.

\bibliographystyle{alpha}
\bibliography{kdst}

\newcommand{\etalchar}[1]{$^{#1}$}
\begin{thebibliography}{FKK{\etalchar{+}}14}

\bibitem[CCC{\etalchar{+}}99]{CharikarCCDGGL99}
Moses Charikar, Chandra Chekuri, To{-}Yat Cheung, Zuo Dai, Ashish Goel, Sudipto
  Guha, and Ming Li.
\newblock Approximation algorithms for directed steiner problems.
\newblock {\em J. Algorithms}, 33(1):73--91, 1999.

\bibitem[CGL15]{ChalermsookGL15}
Parinya Chalermsook, Fabrizio Grandoni, and Bundit Laekhanukit.
\newblock On survivable set connectivity.
\newblock In {\em Proceedings of the Twenty-Sixth Annual {ACM-SIAM} Symposium
  on Discrete Algorithms, {SODA} 2015, San Diego, CA, USA, January 4-6, 2015},
  pages 25--36, 2015.

\bibitem[CLNV14]{CheriyanLNV14}
Joseph Cheriyan, Bundit Laekhanukit, Guyslain Naves, and Adrian Vetta.
\newblock Approximating rooted steiner networks.
\newblock {\em {ACM} Transactions on Algorithms}, 11(2):8:1--8:22, 2014.

\bibitem[Fei98]{Feige98}
Uriel Feige.
\newblock A threshold of ln \emph{n} for approximating set cover.
\newblock {\em J. {ACM}}, 45(4):634--652, 1998.

\bibitem[FKK{\etalchar{+}}14]{FriggstadKKLST14}
Zachary Friggstad, Jochen K{\"{o}}nemann, Young Kun{-}Ko, Anand Louis, Mohammad
  Shadravan, and Madhur Tulsiani.
\newblock Linear programming hierarchies suffice for directed steiner tree.
\newblock In {\em Integer Programming and Combinatorial Optimization - 17th
  International Conference, {IPCO} 2014, Bonn, Germany, June 23-25, 2014.
  Proceedings}, pages 285--296, 2014.

\bibitem[FKN12]{FeldmanKN12}
Moran Feldman, Guy Kortsarz, and Zeev Nutov.
\newblock Improved approximation algorithms for directed steiner forest.
\newblock {\em J. Comput. Syst. Sci.}, 78(1):279--292, 2012.

\bibitem[GKR00]{GargKR00}
Naveen Garg, Goran Konjevod, and R.~Ravi.
\newblock A polylogarithmic approximation algorithm for the group steiner tree
  problem.
\newblock {\em J. Algorithms}, 37(1):66--84, 2000.

\bibitem[HRZ01]{HelvigRZ01}
Christopher~S. Helvig, Gabriel Robins, and Alexander Zelikovsky.
\newblock An improved approximation scheme for the group steiner problem.
\newblock {\em Networks}, 37(1):8--20, 2001.

\bibitem[Lae14]{Laekhanukit14}
Bundit Laekhanukit.
\newblock Parameters of two-prover-one-round game and the hardness of
  connectivity problems.
\newblock In {\em Proceedings of the Twenty-Fifth Annual {ACM-SIAM} Symposium
  on Discrete Algorithms, {SODA} 2014, Portland, Oregon, USA, January 5-7,
  2014}, pages 1626--1643, 2014.

\bibitem[LY94]{LundY94}
Carsten Lund and Mihalis Yannakakis.
\newblock On the hardness of approximating minimization problems.
\newblock {\em J. {ACM}}, 41(5):960--981, 1994.

\bibitem[Rot11]{Rothvoss11}
Thomas Rothvo{\ss}.
\newblock Directed steiner tree and the lasserre hierarchy.
\newblock {\em CoRR}, abs/1111.5473, 2011.

\bibitem[Zel97]{Zelikovsky97}
Alexander Zelikovsky.
\newblock A series of approximation algorithms for the acyclic directed steiner
  tree problem.
\newblock {\em Algorithmica}, 18(1):99--110, 1997.

\end{thebibliography}

\end{document}